\documentclass[a4paper,11pt]{article}
\usepackage[left=3.17cm,right=3.17cm,top=2.54cm,
headheight=0.5cm,headsep=0.54cm,bottom=2.54cm,footskip=0.79cm
]{geometry}

\usepackage[all]{xy}
\usepackage{amsmath,amsfonts,amssymb,amsthm,epsfig,amscd,comment,latexsym,psfrag}
\usepackage{nicefrac,xspace,tikz}

\usetikzlibrary{arrows}
\usetikzlibrary{decorations.markings}
\def\Z{\mathbb Z}

\def\C{\mathbb C}

\def\1{{\bf{1}}}

\usepackage[pdfstartview=FitH]{hyperref}

\def\footnoterule{\kern 1mm \hrule width 7cm \kern 2.2mm}%

 \newtheorem{thm}{Theorem}[section]
 \newtheorem{prp}[thm]{Proposition}
 \newtheorem{lem}[thm]{Lemma}

 \newtheorem{dfn}[thm]{Definition}
 \newtheorem{eg}[thm]{Example}

\newcommand{\bea}{\begin{eqnarray}}
\newcommand{\eea}{\end{eqnarray}}
\newcommand{\be}{\begin{equation}}
\newcommand{\ee}{\end{equation}}

\begin{document}

\title{Algebraic formulas for the structure constants in symmetric functions}
\author{\  \ Na Wang\dag, Ke Wu\ddag\\
\dag\small School of mathematics and statistics, Henan University, Kaifeng, 475001, China.\\
\ddag\small School of Mathematical Sciences, Capital Normal University, Beijing 100048, China.
}

\date{}
\maketitle

\begin{abstract}
Littlewood-Richardson rule gives the decomposition formula for the multiplication of two Schur functions, while the decomposition formula for the multiplication of two Hall-Littlewood functions or two universal characters is also given by the combinatorial method.
In this paper, using the vertex operator realizations of these symmetric functions, we construct the algebraic forms of these decomposition formulas.
\end{abstract}
\noindent
{\bf Keywords: }{Schur function, Hall-Littlewood function, universal character, vertex operator.}

\section{Introduction}\label{sec1}
Symmetric functions have played an important role in mathematics for a long time\cite{Mac,FH,stan,weyl}, ranging from combinatorics, representation theory to enumerative geometry.
Symmetric functions also appear in mathematical physics, especially in integrable models. Kyoto school use Schur functions in a remarkable way to understand the KP and KdV hierarchies\cite{MJD}.
N.V.Tsilevich and P. Su\l kowski give the realizations of phase model and $q$-boson model in the algebra of Schur functions and Hall-Littlewood functions respectively\cite{NVT,PS}. We construct the realization of two-site generalized phase model in the algebra of universal characters (a generalization of Schur functions)\cite{wang1} and find that the generating functions of weighted plane partitions in finite boxes can be obtained from the vertex operators which are raising operators of Hall-Littlewood functions\cite{na1}.

 Schur functions, Hall-Littlewood functions and universal characters can be realized by vertex operators which are defined with the help of infinite dimensional Heisenberg algebras\cite{MJD,Jing}.  Ones can find a general treatment of vertex operators and their connections with affine Lie algebras and the Monster group in \cite{frenkel}.

In the present work, we consider the formulas
\begin{eqnarray*}
S_{\mu}({\bf x})S_{\nu}({\bf x})&=&\sum_\lambda C_{\mu\nu}^\lambda S_{\lambda}({\bf x}),\\
Q_{\mu}({\bf x})Q_{\nu}({\bf x})&=&\sum_\lambda f_{\mu\nu}^\lambda Q_{\lambda}({\bf x}),\\
S_{[\xi,\eta]}({\bf x},{\bf y})S_{[\tau,\mu]}({\bf x},{\bf y})&=&\sum_{\lambda,\mu} M_{[\xi,\eta],[\tau,\nu]}^{[\lambda,\mu]}S_{[\lambda,\mu]}({\bf x},{\bf y}),
\end{eqnarray*} which are given by the combinatorial method, where $S_{\lambda}({\bf x}),\ Q_{\lambda}({\bf x})$ and $ S_{[\lambda,\mu]}({\bf x},{\bf y})$ denote Schur function, Hall-Littlewood function and universal character respectively. we construct the algebraic formulas to compute them. All results we obtain are based on the vertex operator realizations of these three kinds of symmetric functions.

For an integer vector $\alpha=(\alpha_1,\alpha_2,\cdots,\alpha_l)$, we define the polynomials $S_\alpha({\bf x})$ by \[
V^{+}(z_1)V^{+}(z_2)\cdots V^{+}(z_l)\cdot 1=\sum_{\alpha}S_\alpha({\bf x})z_1^{\alpha_1}z_2^{\alpha_2}\cdots z_l^{\alpha_l},\]
where
\[
V^+(z)=e^{ \sum_{n=1}^\infty x_n z^n}e^{- \sum_{n=1}^\infty\frac{1}{n}\frac{\partial}{ \partial{x_n} }z^{-n}}.
\]
When $\alpha$ is a Young diagram, then $S_\alpha({\bf x})$ is a Schur function. When $\alpha$ is not a Young diagram, then $S_\alpha({\bf x})$ can turn into a Schur function by the following formula
\[S_{(\alpha_1,\cdots,\alpha_i,\alpha_{i+1},\cdots,\alpha_l)}({\bf x})=-S_{(\alpha_1,\cdots,\alpha_{i+1}-1,\alpha_{i}+1,\cdots,\alpha_l)}({\bf x}).
\]

For a pair of Young diagram $\mu=(\mu_1,\mu_2,\cdots,\mu_l)$ and $\nu=(\nu_1,\nu_2,\cdots,\nu_{l'})$, the Schur functions $S_\mu({\bf x})$ and $S_\nu({\bf x})$, respectively,  equal the coefficients of $z_1^{\mu_1}\cdots z_l^{\mu_l}$ and $w_1^{\nu_1}\cdots w_{l'}^{\nu_{l'}}$ in the expansions of
\[V^{+}(z_1)V^{+}(z_2)\cdots V^{+}(z_l)\cdot 1,\ \quad \text{and}\ \quad V^{+}(w_1)V^{+}(w_2)\cdots V^{+}(w_{l'})\cdot 1.\]

By calculation, the multiplication
\begin{eqnarray*}
&& V^{+}(z_1)V^{+}(z_2)\cdots V^{+}(z_l)\cdot 1 \times V^{+}(w_1)V^{+}(w_2)\cdots V^{+}(w_l)\cdot 1\\
&=&\prod_{i,j}\frac{1}{1-w_j/z_i}V^{+}(z_1)V^{+}(z_2)\cdots V^{+}(z_l)V^{+}(w_1)V^{+}(w_2)\cdots V^{+}(w_l)\cdot 1,
\end{eqnarray*}
taking the coefficient of $z_1^{\mu_1}\cdots z_l^{\mu_l}w_1^{\nu_1}\cdots w_{l'}^{\nu_{l'}}$ of the both sides, we get the algebraic formula to calculate the multiplication $S_\mu({\bf x})S_\nu({\bf x})$.
With the same method, we get the the algebraic formula to calculate the multiplication of two Hall-Littlewood functions.

The paper is organized as follows. In section \ref{sect2}, we get the algebraic form of the decomposition formula of the multiplication of two Schur functions.  In section \ref{sect3}, we get the algebraic form of the decomposition formula of the multiplication of two Hall-Littlewood functions. In section \ref{sect4}, we get the the algebraic form of the decomposition formula of the multiplication of two universal characters.
\section{Schur functions}\label{sect2}
Let ${\bf x}=(x_1,x_2,\cdots)$. The operators $h_n({\bf x})$ are determined by the generated function:
\be\label{hxi}
\sum_{n=0}^\infty h_n({\bf x})z^n=e^{\xi({\bf x},z)},\quad \xi({\bf x},z)=\sum_{n=1}^\infty x_n z^n
\ee
and set $h_n({\bf x})=0$ for $n<0$. The operators $h_n({\bf x})$ can be explicitly written as
\[
h_n({\bf x})=\sum_{k_1+2k_2+\cdots nk_n=n}\frac{x_1^{k_1}x_2^{k_2}\cdots x_n^{k_n}}{k_1!k_2!\cdots k_n!}
\]
Note that if we replace $ix_i$ with the power sum $p_i=\sum_j x_j^i$, $h_n({\bf x})$ is the complete homogeneous symmetric function
\[
\sum_{i_1\leq\cdots\leq i_n} x_{i_1}x_{i_2}\cdots x_{i_n}.
\]

For Young diagram $\lambda=(\lambda_1,\lambda_2,\cdots,\lambda_l)$, Schur function $S_{\lambda}({\bf x})$ is a polynomial of variables ${\bf x}$ in $\C[{\bf x}]$ defined by the Jacobi-Trudi formula \cite{Mac}:
\be
S_{\lambda}({\bf x})=\det\left(h_{\lambda_{i}-i+j}({\bf x})\right)_{1\leq i\leq l}
\ee

Introduce the vertex operators
\begin{eqnarray}
V^{\pm}(z)=\sum_{n\in\Z}V_n^{\pm}z^n=e^{\pm \xi({\bf x},z)}e^{\mp \xi(\tilde\partial_{\bf x},z^{-1})},
\end{eqnarray}
where $\tilde\partial_{\bf x}=(\frac{\partial}{\partial{{\bf x}_1}},\frac{1}{2}\frac{\partial}{\partial{{\bf x}_2}},\cdots)$.
The operators $V_n^+$ are raising operators of the Schur function in the following sense:
\begin{eqnarray}
S_\lambda({\bf x})&=&V_{\lambda_1}^+V_{\lambda_2}^+\cdots V_{\lambda_l}^+\cdot1\\
&=&[{\bf z}^\lambda]V^{+}(z_1)V^{+}(z_2)\cdots V^{+}(z_l)\cdot 1\label{slambda}
\end{eqnarray}
where $[{\bf z}^\lambda]=[z_1^{\lambda_1}z_2^{\lambda_2}\cdots z_l^{\lambda_l}]$ means taking the coefficient of the term $z_1^{\lambda_1}z_2^{\lambda_2}\cdots z_l^{\lambda_l}$ in the expansion of $V^{+}(z_1)V^{+}(z_2)\cdots V^{+}(z_l)\cdot 1$.

\begin{dfn}
For an integer vector $\alpha=(\alpha_1,\alpha_2,\cdots,\alpha_l)$, define the polynomials $S_\alpha({\bf x})$ by \begin{eqnarray}
S_\alpha({\bf x})&=&\det\left(h_{\alpha_{i}-i+j}({\bf x})\right)_{1\leq i,j\leq l}\nonumber\\
&=&V_{\alpha_1}^+V_{\alpha_2}^+\cdots V_{\alpha_l}^+\cdot1\nonumber\\
&=&[{\bf z}^\alpha]V^{+}(z_1)V^{+}(z_2)\cdots V^{+}(z_l)\cdot 1\label{slambda}
\end{eqnarray}
where $[{\bf z}^\alpha]=[z_1^{\alpha_1}z_2^{\alpha_2}\cdots z_l^{\alpha_l}]$. We denote $S_\alpha({\bf x})$ by $S_\alpha$ for simplicity.
\end{dfn}
\begin{lem}\label{al12}
The polynomials $S_\alpha({\bf x})$ defined in (\ref{slambda}) equal zero unless $\alpha_1\geq -l+1,\alpha_2\geq -l+2,\cdots,\alpha_l\geq 0$.
\end{lem}
\begin{proof}Since \be
V^{+}(z_1)V^{+}(z_2)\cdots V^{+}(z_l)\cdot 1=\prod_{i<j}(1-\frac{z_j}{z_i})e^{\xi({\bf x},z_1)}e^{\xi({\bf x},z_2)}\cdots e^{\xi({\bf x},z_l)},
\ee
taking the coefficient of $z^\alpha$, we get the result.
\end{proof}
\begin{lem}
If $\alpha_j=\alpha_i+(j-i)$, the polynomial $S_\alpha({\bf x})$ is equal to $0$.
\end{lem}
\begin{proof}
This holds since the $i$th line and the $j$th line of $\det\left(h_{\alpha_{i}-i+j}({\bf x})\right)_{1\leq i,j\leq l}$ are the same if $\alpha_j=\alpha_i+(j-i)$.
\end{proof}
Note that if $\alpha$ is a Young diagram, i.e., $\alpha_1\geq\alpha_2\geq\cdots\geq\alpha_l\geq 0$, the polynomial $S_\alpha({\bf x})$ is a Schur function corresponding to the Young diagram $\alpha$. If $\alpha$ is not a Young diagram, the polynomial $S_\alpha({\bf x})$ can be turned into a Schur function by the following formula.
\begin{prp}\label{al1221}
For $1\leq i\leq l-1$, the following equation holds
\be\label{eqal1221}
S_{(\alpha_1,\cdots,\alpha_i,\alpha_{i+1},\cdots,\alpha_l)}=-S_{(\alpha_1,\cdots,\alpha_{i+1}-1,\alpha_{i}+1,\cdots,\alpha_l)}.
\ee
\end{prp}
\begin{proof}Since
\begin{eqnarray*}
&&V^{+}(z_1)\cdots V^{+}(z_i)V^+(z_{i+1})\cdots V^{+}(z_l)\cdot 1\\
&=&-\frac{z_{i+1}}{z_i}V^{+}(z_1)\cdots V^{+}(z_{i+1})V^+(z_{i})\cdots V^{+}(z_l)\cdot 1,\end{eqnarray*}
taking the coefficient of $z^\alpha$, we get the result.\end{proof}
For example, $S_{(2,3)}({\bf x})=0$, and $S_{(2,4)}({\bf x})=-S_{(3,3)}({\bf x})$.

\begin{lem}\label{vv}
The vertex operators satisfy
\be\label{zwzw}
\frac{1}{1-w/z}V^+(z)V^+(w)=\frac{1}{1-z/w}V^+(w)V^+(z)
\ee
\begin{proof}
Since
\[
e^{-\xi(\partial_{\bf x},z^{-1})}e^{\xi({\bf x},w)}=(1-\frac{w}{z})e^{\xi({\bf x},w)}e^{-\xi(\partial_{\bf x},z^{-1})}
\]
we have that both sides of equation (\ref{zwzw}) equal
\[
:V^+(z)V^+(w):
\]
where $::$ denotes the normal order defined as usual.
\end{proof}
\end{lem}
For each pair of integers $i, j$ such that $1\leq i<j\leq l$, define $R_{ij}$ by
\be
R_{ij}\cdot \alpha=(\alpha_1,\cdots,\alpha_i+1,\cdots,\alpha_j-1,\cdots,\alpha_l)
\ee
and
\be
R_{ij}\cdot S_{\alpha}({\bf x})= S_{R_{ij}\cdot\alpha}({\bf x}).
\ee

For a pair of Young diagrams $\lambda, \mu$, the notation $\lambda\cup\mu$ is the Young diagram whose parts are those of $\lambda$ and $\mu$ arranged in descending order. For example, if $\lambda=(321)$ and $\mu=(21)$, then $\lambda\cup\mu=(32211)$.

The Pieri's formula \cite{Mac} tells $S_\mu({\bf x}) h_r({\bf x})=\sum_{\lambda}S_\lambda({\bf x})$ summed over all partitions $\lambda$ such that $\lambda-\mu$ is a horizontal $r$-strip. Here we do not explain the notations, but we will give an algebraic formula for $S_\mu({\bf x}) h_r({\bf x})$.

\begin{prp}
Let Young diagram $\mu=(\mu_1,\cdots,\mu_l)$ and $\mu_{i+1}< r\leq\mu_i$ for $0\leq i\leq l$ with $\mu_{l+1}=0,\ \mu_{0}=+\infty$,
\be
S_\mu({\bf x}) h_r({\bf x})=\prod_{j=1}^i\frac{1}{1-R_{j(i+1)}}\prod_{k=i+2}^{l+1}\frac{1}{1-R_{(i+1)k}}\cdot S_{\mu\cup(r)}({\bf x}).
\ee
\end{prp}
\begin{proof}
From equation (\ref{slambda}) and lemma \ref{vv},
\begin{eqnarray*}
S_\mu({\bf x}) h_r({\bf x})&=&[z_1^{\mu_1}\cdots z_l^{\mu_l}]V^+(z_1)\cdots V^+(z_l)\cdot 1 \ [w^r]V^+(w)\cdot 1\\
&=&[z_1^{\mu_1}\cdots z_i^{\mu_i}w^r z_{i+1}^{\mu_{i+1}}\cdots z_l^{\mu_l}]\prod_{j=1}^i\frac{1}{1-w/z_j}\prod_{j=i+1}^l\frac{1}{1-z_j/w}\\
&&V^+(z_1)\cdots V^+(z_i)V^+(w) V^+( z_{i+1})\cdots V^+(z_l)\cdot 1\\
&=&\prod_{j=1}^i\frac{1}{1-R_{j(i+1)}}\prod_{k=i+2}^{l+1}\frac{1}{1-R_{(i+1)k}}\cdot S_{\mu\cup(r)}({\bf x}).
\end{eqnarray*}
\end{proof}
\begin{eg} Let $\mu=(2,1)$ and $r=2$,
\begin{eqnarray*}
S_{(2,1)}({\bf x})h_2({\bf x})&=&\frac{1}{1-R_{12}}\frac{1}{1-R_{23}}\cdot S_{(2,2,1)}({\bf x})\\
&=& (1+R_{12}+R_{12}^2+R_{12}^3+R_{12}^4)(1+R_{23})\cdot S_{(2,2,1)}({\bf x})\\
&=&S_{(2,2,1)}({\bf x})+S_{(3,1,1)}({\bf x})+S_{(4,0,1)}({\bf x})+S_{(5,-1,1)}({\bf x})+S_{(6,-2,1)}({\bf x})\\
&&S_{(2,3)}({\bf x})+S_{(3,2)}({\bf x})+S_{(4,1)}({\bf x}+S_{(5)}({\bf x})+S_{(6,-1)}({\bf x})\\
&=& S_{(2,2,1)}({\bf x})+S_{(3,1,1)}({\bf x})+S_{(3,2)}({\bf x})+S_{(4,1)}({\bf x})
\end{eqnarray*}
where the second equation holds since $(\mu\cup(r))_2=2$ and $(\mu\cup(r))_3=1$, the last equation holds by using the results in lemma \ref{al12} and proposition \ref{al1221}, i.e., by lemma \ref{al12}, $S_{(6,-2,1)}({\bf x})=0$, by
proposition \ref{al1221}, $S_{(4,0,1)}({\bf x})=0,\ S_{(5,-1,1)}({\bf x})=-S_{(5)}({\bf x}),\ S_{(2,3)}({\bf x})=0,\ S_{(6,-1)}({\bf x})=0$.\end{eg}

For a pair of Young diagrams $\mu$ and $\nu$, the product $S_\mu({\bf x})S_\nu({\bf x})$ is an integral linear combination of Schur functions
\be
S_\mu({\bf x})S_\nu({\bf x})=\sum_\lambda c_{\mu\nu}^\lambda S_\lambda({\bf x}).
\ee
The Littlewood-Richardson rule tells $c_{\mu\nu}^\lambda$ is equal to the number of tableaux $T$ of shape $\lambda-\mu$ and weight $\nu$ such that $\omega (T)$ is a lattice permutation. If interested, ones can find a detailed explanation of the Littlewood-Richardson rule in the book \cite{Mac}. In the following, we will give the algebraic form of the decomposition formula for $S_\mu({\bf x})S_\nu({\bf x})$.

\begin{prp}
Let $\mu=(\mu_1,\mu_2,\cdots,\mu_l)$ and $\nu=(\nu_1,\nu_2,\cdots,\nu_{l'})$. Suppose $\mu_{i_k}\geq \nu_k>\mu_{i_k+1}$ for $k=1,2,\cdots,l'$ with $\mu_0=+\infty,\ \mu_{l+1}=0$. Denote the set $\{i_1+1,i_2+2,\cdots, i_{l'}+l'\}$ by $I$, then the following formula holds
\be
S_\mu({\bf x})S_\nu({\bf x})=\prod_{k=1}^{l'}\prod_{j=1,j\not\in I}^{i_k+k-1}\frac{1}{1-R_{j(i_k+k)}}\prod_{j=i_k+k+1,j\not\in I}^{l+l'}\frac{1}{1-R_{j(i_k+k)}}\cdot S_{\mu\cup\nu}({\bf x}).
\ee
\end{prp}
\begin{proof}
From equation (\ref{slambda}) and lemma \ref{vv},
\begin{eqnarray*}
&&S_\mu({\bf x}) S_\nu({\bf x})=[{\bf z}^\lambda]V^+(z_1)\cdots V^+(z_l)\cdot 1 \ [{\bf w}^\mu]V^+(w_1)\cdots V^+(w_{l'})\cdot 1\\
&=&[{\bf z}^\lambda{\bf w}^\mu]\prod_{k=1}^{l'}\prod_{j=1}^{i_k}\frac{1}{1-w_{i_k+k}/z_j}\prod_{j=i_k+1}^{l}\frac{1}{1-z_{j}/w_{i_k+k}}
V^+(z_1)\cdots V^+(z_{i_1})\\&&V^+(w_1) V^+( z_{i_1+1})\cdots V^+(z_{i_{l'}})V^+(w_{l'}) V^+( z_{i_{1'+1}})\cdots V^+(z_l)\cdot 1\\
&=&\prod_{k=1}^{l'}\prod_{j=1,j\not\in I}^{i_k+k-1}\frac{1}{1-R_{j(i_k+k)}}\prod_{j=i_k+k+1,j\not\in I}^{l+l'}\frac{1}{1-R_{j(i_k+k)}}\cdot S_{\mu\cup\nu}({\bf x}).
\end{eqnarray*}
\end{proof}
\begin{eg}Let $\mu=\nu=(2,1)$,
\begin{eqnarray*}
S_\mu({\bf x}) S_\nu({\bf x})&=&\frac{1}{1-R_{12}}\frac{1}{1-R_{23}}\frac{1}{1-R_{14}}\frac{1}{1-R_{34}}\cdot S_{(2,2,1,1)}({\bf x})\\
&=& \frac{1}{1-R_{12}}\frac{1}{1-R_{23}}\left(S_{(2,2,1,1)}+S_{(3,2,1)}+S_{(2,2,2)}\right)\\
&=& \frac{1}{1-R_{12}}\left(S_{(2,2,1,1)}+S_{(2,2,2)}+S_{(3,2,1)}+S_{(3,3)}+S_{(2,3,1)}\right)\\
&=&S_{(2,2,1,1)}({\bf x})+S_{(3,3)}({\bf x})+2S_{(3,2,1)}({\bf x})+S_{(2,2,2)}({\bf x})\\
&&+S_{(3,1,1,1)}({\bf x})+S_{(4,1,1)}({\bf x})+S_{(4,2)}({\bf x}).
\end{eqnarray*}
\end{eg}
\section{Hall-Littlewood functions}\label{sect3}
The operators $q_n({\bf x})$ are determined by the generated function\cite{Jing,FW}:
\be\label{qr}
\sum_{n=0}^\infty q_n({\bf x})z^n=e^{\xi_t({\bf x},z)},\quad \xi_t({\bf x},z)=\sum_{n=1}^\infty(1-t^n) x_n z^n
\ee
and set $q_n({\bf x})=0$ for $n<0$. When $t=0$, the operators $q_n({\bf x})$ turn into the complete homogeneous symmetric function $h_n({\bf x})$ if we replace $ix_i$ with the power sum $p_i$.

For a Young diagram $\lambda=(\lambda_1,\lambda_2,\cdots,\lambda_l)$, define\begin{equation}
q_\lambda({\bf x})=q_{\lambda_1}({\bf x})q_{\lambda_2}({\bf x})\cdots q_{\lambda_l}({\bf x})
\end{equation}
and the Hall-Littlewood function $Q_\lambda({\bf x})$ equals
\be
Q_\lambda({\bf x})=\prod_{i<j}\frac{1-R_{ij}}{1-tR_{ij}}q_\lambda({\bf x}).
\ee

The Hall-Littlewood function $P_\lambda({\bf x})$ is a scalar multiple of $Q_\lambda({\bf x})$:
\be
P_\lambda({\bf x})=\frac{1}{b_\lambda(t)}Q_\lambda({\bf x})
\ee
where
\be
b_\lambda(t)=\prod_{i\geq 1}\phi_{m_i(\lambda)}(t).
\ee
Here $m_i(\lambda)$ denotes the number of times $i$ occurs as a part of $\lambda$, and $\phi_r(t)=(1-t)(1-t^2)\cdots(1-t^r)$.

Introduce the following vertex operators
\bea
\Gamma^-(z)=e^{\xi_t({\bf x}, z)},\quad \Gamma^+(z)=e^{\xi(\tilde{\partial}_{\bf x}, z^{-1})},\label{g1}
\eea
where $\xi({\bf x}, z)$ is the special case $t=0$ of $\xi_t({\bf x}, z)$.
Define
\bea
X^\pm(z)&=&\sum_{n\in\Z}X^\pm_nz^n=e^{\pm\xi_t({\bf x}, z)}e^{\mp\xi(\tilde{\partial}_{\bf x}, z^{-1})}.
\eea
The operators $X_i^\pm$ satisfy the following deformed fermionic relations:
\bea
X_{n-1}^\pm X_m^\pm-tX_{n}^\pm X_{m-1}^\pm +X_{m-1}^\pm X_n^\pm-tX_{m}^\pm X_{n-1}^\pm&=&0,\nonumber\\
X_n^+X_{m-1}^--t X_{n-1}^+X_{m}^-+X_m^+X_{n-1}^--t X_{m-1}^+X_{n}^-&=&(1-t)^2\delta_{m+n,1}.\nonumber
\eea
The operators $X_i^+$ are raising operators for the Hall-Littlewood functions\cite{Jing} such that
\begin{equation}\label{qdef}
Q_{\lambda}({\bf x})=X_{\lambda_1}^+X_{\lambda_2}^+\cdots X_{\lambda_l}^+\cdot 1=[z^\lambda]X^+(z_1)X^+(z_2)\cdots X^+(z_l)\cdot 1
\end{equation}
where the Young diagrams $\lambda=(\lambda_1,\lambda_2,\cdots,\lambda_l)$.

\begin{dfn}
For an integer vector $\alpha=(\alpha_1,\alpha_2,\cdots,\alpha_l)$, define the polynomials $Q_\alpha({\bf x})$ by \begin{eqnarray}
Q_\alpha({\bf x})&=&\prod_{i<j}\frac{1-R_{ij}}{1-tR_{ij}}q_{\alpha}({\bf x})\nonumber\\
&=&X_{\alpha_1}^+X_{\alpha_2}^+\cdots X_{\alpha_l}^+\cdot1\nonumber\\
&=&[{\bf z}^\alpha]X^{+}(z_1)X^{+}(z_2)\cdots X^{+}(z_l)\cdot 1\label{qlambda}
\end{eqnarray}
where $[{\bf z}^\alpha]=[z_1^{\alpha_1}z_2^{\alpha_2}\cdots z_l^{\alpha_l}]$ and $q_{\alpha}({\bf x})=q_{\alpha_1}({\bf x})q_{\alpha_2}({\bf x})\cdots q_{\alpha_l}({\bf x})$. We denote $Q_\alpha({\bf x})$ by $Q_\alpha$ for short.
\end{dfn}
\begin{lem}\label{qal12}
The polynomials $Q_\alpha({\bf x})$ defined in (\ref{qlambda}) equal zero unless
\begin{eqnarray*}
\alpha_l&\geq& 0,\\ \alpha_{l-1}+\alpha_l&\geq &0,\\
\alpha_{l-2}+\alpha_{l-1}+\alpha_l&\geq& 0,\\
&\cdots,&\\
\alpha_{1}+\alpha_{2}+\cdots+\alpha_l&\geq& 0.
\end{eqnarray*}
\end{lem}
\begin{proof}Since \be
X^{+}(z_1)X^{+}(z_2)\cdots X^{+}(z_l)\cdot 1=\prod_{i<j}\frac{1-{z_j}/{z_i}}{1-tz_j/z_i}e^{\xi_t({\bf x},z_1)}e^{\xi_t({\bf x},z_2)}\cdots e^{\xi_t({\bf x},z_l)},
\ee
taking the coefficient of $z^\alpha$, we get the result.
\end{proof}
Note that if $\alpha$ is a Young diagram, i.e., $\alpha_1\geq\alpha_2\geq\cdots\geq\alpha_l\geq 0$, the polynomial $Q_\alpha({\bf x})$ is a Hall-Littlewood function corresponding to the Young diagram $\alpha$. If $\alpha$ is not a Young diagram, the polynomial $Q_\alpha({\bf x})$ equals the sum of some Hall-Littlewood functions by the following formula.
\begin{prp}\label{al1221}
For $1\leq i\leq l-1$, the following equation holds
\be\label{alii+1}
Q_{(\alpha_1,\cdots,\alpha_i,\alpha_{i+1},\cdots,\alpha_l)}=\frac{t-R_{i(i+1)}}{1-tR_{i(i+1)}}Q_{(\alpha_1,\cdots,\alpha_{i+1},\alpha_{i},\cdots,\alpha_l)},
\ee
where
\[
R_{i(i+1)}\cdot Q_{(\alpha_1,\cdots,\alpha_{i+1},\alpha_{i},\cdots,\alpha_l)}=Q_{(\alpha_1,\cdots,\alpha_{i+1}-1,\alpha_{i}+1,\cdots,\alpha_l)}.
\]
\end{prp}
\begin{proof}Since
\begin{eqnarray*}
&&X^{+}(z_1)\cdots X^{+}(z_i)X^+(z_{i+1})\cdots X^{+}(z_l)\cdot 1\\
&=&\frac{t-z_{i+1}/z_i}{1-tz_{i+1}/z_i}X^{+}(z_1)\cdots X^{+}(z_{i+1})X^+(z_{i})\cdots X^{+}(z_l)\cdot 1,\end{eqnarray*}
taking the coefficient of $z^\alpha$, we get the result.\end{proof}

\begin{eg}\label{eg23} We calculate $Q_{(2,3)}({\bf x})$, $Q_{(1,4)}({\bf x})$, $Q_{(0,5)}({\bf x})$, $Q_{(-1,6)}({\bf x})\cdots$.

From (\ref{alii+1}) and lemma \ref{qal12}, we have
\begin{eqnarray*}
Q_{(2,3)}&=&tQ_{(3,2)}+(t^2-1)Q_{(2,3)}+t(t^2-1)Q_{(1,4)}+t^2(t^2-1)Q_{(0,5)}\nonumber\\
&&+t^3(t^2-1)Q_{(-1,6)}+t^4(t^2-1)Q_{(-2,7)}+\cdots\\
Q_{(1,4)}&=&tQ_{(4,1)}+(t^2-1)Q_{(3,2)}+t(t^2-1)Q_{(2,3)}+t^2(t^2-1)Q_{(1,4)}\nonumber\\
&&+t^3(t^2-1)Q_{(0,5)}+t^4(t^2-1)Q_{(-1,6)}+t^5(t^2-1)Q_{(-2,7)}+\cdots\\
Q_{(0,5)}&=&tQ_{(5)}+(t^2-1)Q_{(4,1)}+t(t^2-1)Q_{(3,2)}+t^2(t^2-1)Q_{(2,3)}\nonumber\\
&&+t^3(t^2-1)Q_{(1,4)}+t^4(t^2-1)Q_{(0,5)}+t^5(t^2-1)Q_{(-1,6)}+\cdots\\
Q_{(-1,6)}&=&(t^2-1)Q_{(5)}+t(t^2-1)Q_{(4,1)}+t^2(t^2-1)Q_{(3,2)}+t^3(t^2-1)Q_{(2,3)}\nonumber\\
&&+t^4(t^2-1)Q_{(1,4)}+t^5(t^2-1)Q_{(0,5)}+t^6(t^2-1)Q_{(-1,6)}+\cdots\\
&\cdots&
\end{eqnarray*}
we get
\begin{eqnarray*}
Q_{(1,4)}-tQ_{(2,3)}&=&tQ_{(4,1)}-Q_{(3,2)},\\
Q_{(0,5)}-t^2Q_{(2,3)}&=&t Q_{(5)}+(t^2-1)Q_{(4,1)}-tQ_{(3,2)},\\
Q_{(-1,6)}-t^3Q_{(2,3)}&=&(t^2-1) Q_{(5)}+t(t^2-1)Q_{(4,1)}-t^2Q_{(3,2)},\\
Q_{(-2,7)}=tQ_{(-1,6)},&& Q_{(-3,8)}=t^2 Q_{(-1,6)},\ \ Q_{(-4,9)}=t^3 Q_{(-1,6)},\cdots
\end{eqnarray*}
then
\begin{eqnarray}
Q_{(2,3)}&=&tQ_{(3,2)},\\
Q_{(1,4)}&=&tQ_{(4,1)}+(t^2-1)Q_{(3,2)},\nonumber\\
Q_{(0,5)}&=&tQ_{(5)}+(t^2-1)Q_{(4,1)}+t(t^2-1)Q_{(3,2)},\nonumber\\
Q_{(-1,6)}&=&(t^2-1)Q_{(5)}+t(t^2-1)Q_{(4,1)}+t^2(t^2-1)Q_{(3,2)},\nonumber\\
Q_{(-2,7)}&=& tQ_{(-1,6)},\ \ Q_{(-3,8)}\ =\ t^2 Q_{(-1,6)},\ \ Q_{(-4,9)}\ =\ t^3 Q_{(-1,6)},\cdots.\nonumber
\end{eqnarray}
\end{eg}

\begin{eg}\label{eg24}
We calculate $Q_{(2,4)}({\bf x})$, $Q_{(1,5)}({\bf x})$, $Q_{(0,6)}({\bf x})$, $Q_{(-1,7)}({\bf x})\cdots$.

From (\ref{alii+1}) and lemma \ref{qal12}, we have
\begin{eqnarray*}
Q_{(2,4)}&=&tQ_{(4,2)}+(t^2-1)Q_{(3,3)}+t(t^2-1)Q_{(2,4)}+t^2(t^2-1)Q_{(1,5)}+t^3(t^2-1)Q_{(0,6)}\nonumber\\
&&+t^4(t^2-1)Q_{(-1,7)}+\cdots,\nonumber\\
Q_{(1,5)}&=&tQ_{(5,1)}+(t^2-1)Q_{(4,2)}+t(t^2-1)Q_{(3,3)}+t^2(t^2-1)Q_{(2,4)}+t^3(t^2-1)Q_{(1,5)}\nonumber\\
&&+t^4(t^2-1)Q_{(0,6)}+t^5(t^2-1)Q_{(-1,7)}+\cdots,\nonumber\\
Q_{(0,6)}&=&tQ_{(6)}+(t^2-1)Q_{(5,1)}+t(t^2-1)Q_{(4,2)}+t^2(t^2-1)Q_{(3,3)}+t^3(t^2-1)Q_{(2,4)}\nonumber\\
&&+t^4(t^2-1)Q_{(1,5)}+t^5(t^2-1)Q_{(0,6)}+t^6(t^2-1)Q_{(-1,7)}+\cdots\\
Q_{(-1,7)}&=&(t^2-1)Q_{(6)}+t(t^2-1)Q_{(5,1)}+t^2(t^2-1)Q_{(4,2)}+t^3(t^2-1)Q_{(3,3)}\nonumber\\
&&+t^4(t^2-1)Q_{(2,4)}+t^5(t^2-1)Q_{(1,5)}+t^6(t^2-1)Q_{(0,6)}+t^7(t^2-1)Q_{(-1,7)}+\cdots\\
Q_{(-2,8)}&=&tQ_{(-1,7)},\ \ Q_{(-3,9)}\ =\ t^2Q_{(-1,7)},\ \cdots.
\end{eqnarray*}
we get
\begin{eqnarray*}
Q_{(1,5)}-tQ_{(2,4)}&=&tQ_{(5,1)}-Q_{(4,2)},\\
Q_{(0,6)}-t^2Q_{(2,4)}&=&t Q_{(6)}+(t^2-1)Q_{(5,1)}-tQ_{4,2},\\
Q_{(-1,7)}-t^3Q_{(2,4)}&=&(t^2-1) Q_{(6)}+t(t^2-1)Q_{(5,1)}-t^2Q_{(4,2)},
\end{eqnarray*}
then
\begin{eqnarray*}
Q_{(2,4)}&=&tQ_{(4,2)}+(t-1)Q_{(3,3)},\\
Q_{(1,5)}&=&tQ_{(5,1)}+(t^2-1)Q_{(4,2)}+t(t-1)Q_{(3,3)},\nonumber\\
Q_{(0,6)}&=&tQ_{(6)}+(t^2-1)Q_{(5,1)}+t(t^2-1)Q_{(4,2)}++t^2(t-1)Q_{(3,3)},\nonumber\\
Q_{(-1,7)}&=&(t^2-1)Q_{(6)}+t(t^2-1)Q_{(5,1)}+t^2(t^2-1)Q_{(4,2)}++t^3(t-1)Q_{(3,3)},\nonumber\\
Q_{(-2,8)}&=&tQ_{(-1,7)},\ \ Q_{(-3,9)}\ =\ t^2Q_{(-1,7)},\ \cdots.
\end{eqnarray*}
\end{eg}

For a polynomial $Q_{\alpha}({\bf x})=Q_{(\alpha_1,\cdots,\alpha_i,\alpha_{i+1},\cdots,\alpha_l)}({\bf x})$, we want to change it to the form $Q_{(\alpha_1,\cdots,\alpha_{i+1},\alpha_{i},\cdots,\alpha_l)}({\bf x})$ if $\alpha_i<\alpha_{i+1}$. From (\ref{alii+1}), We get two different formulas based on $\alpha_{i+1}-\alpha_{i}$ equal to an odd or even number.

\begin{prp}\label{prp23}
If $\alpha_{i+1}-\alpha_{i}$ are odd, the polynomials $Q_{(\alpha_1,\cdots,\alpha_i,\alpha_{i+1},\cdots,\alpha_l)}({\bf x})$ are calculated by the following formulas:
\begin{eqnarray*}
Q_{(\cdots, n, n+1,\cdots)}&=&tQ_{(\cdots, n+1, n,\cdots)},\\
Q_{(\cdots, n-1, n+2,\cdots)}&=&tQ_{(\cdots, n+2, n-1,\cdots)}+(t^2-1)Q_{(\cdots, n+1, n,\cdots)},\\
Q_{(\cdots, n-2, n+3,\cdots)}&=&tQ_{(\cdots, n+3, n-2,\cdots)}+(t^2-1)Q_{(\cdots, n+2, n-1,\cdots)}+t(t^2-1)Q_{(\cdots, n+1, n,\cdots)},\\
&\cdots&\\
Q_{(\cdots, n-m, n+1+m,\cdots)}&=&tQ_{(\cdots, n+1+m, n-m,\cdots)}+(t^2-1)Q_{(\cdots, n+m, n-m+1,\cdots)}+\cdots\\
&&+t^{m-1}(t^2-1)Q_{(\cdots, n+1, n,\cdots)},
\end{eqnarray*}
where $n,m\in\Z$ and $m>0$.
\end{prp}
\begin{proof}
These equations can be proved by the method used in example \ref{eg23}.
\end{proof}

\begin{prp}\label{prp24}
If $\alpha_{i+1}-\alpha_{i}$ are even, the polynomials $Q_{(\alpha_1,\cdots,\alpha_i,\alpha_{i+1},\cdots,\alpha_l)}({\bf x})$ are calculated by the following formulas:
\begin{eqnarray*}
Q_{(\cdots, n, n+2,\cdots)}&=&tQ_{(\cdots, n+2, n,\cdots)}+(t-1)Q_{(\cdots, n+1, n+1,\cdots)},\\
Q_{(\cdots, n-1, n+3,\cdots)}&=&tQ_{(\cdots, n+3, n-1,\cdots)}+(t^2-1)Q_{(\cdots, n+2, n,\cdots)}+t(t-1)Q_{(\cdots, n+1, n+1,\cdots)},\\
Q_{(\cdots, n-2, n+4,\cdots)}&=&tQ_{(\cdots, n+4, n-2,\cdots)}+(t^2-1)Q_{(\cdots, n+3, n-1,\cdots)}+t(t^2-1)Q_{(\cdots, n+2, n,\cdots)}\\
&&+t^2(t-1)Q_{(\cdots, n+1, n+1,\cdots)},\\
&\cdots&\\
Q_{(\cdots, n-m, n+2+m,\cdots)}&=&tQ_{(\cdots, n+2+m, n-m,\cdots)}+(t^2-1)Q_{(\cdots, n+1+m, n-m+1,\cdots)}+\cdots\\
&&+t^{m-1}(t^2-1)Q_{(\cdots, n+2, n,\cdots)}+t^m(t-1)Q_{(\cdots, n+1, n+1,\cdots)},
\end{eqnarray*}
where $n,m\in\Z$ and $m>0$.
\end{prp}
\begin{proof}
These equations can be proved by the method used in example \ref{eg24}.
\end{proof}

Note that the formulas in propositions \ref{prp23} and \ref{prp24} turn into the equation (\ref{eqal1221}) when $t=0$.

\begin{lem}\label{Qqlem} The following equation holds
\be
\frac{1-tw/z}{1-w/z}X^+(z)X^+(w)=\frac{1-tz/w}{1-z/w}X^+(w)X^+(z).
\ee
\end{lem}

For a Young diagram $\mu$,
\be\label{qQ}
q_r({\bf x}) Q_\mu({\bf x})=\sum_\lambda \psi_{\lambda/\mu}(t) Q_\lambda({\bf x})
\ee
summed over all $\lambda\supset\mu$ such that $\lambda-\mu$ is a horizontal $r$-strip\cite{Mac}. Here
\be\label{psi11}
\psi_{\lambda/\mu}(t)=\prod_{j\in J_{\lambda/\mu}}(1-t^{m_j(\mu)})
\ee
where $\theta=\lambda-\mu$ is a horizontal strip and $J_{\lambda/\mu}$ is the set of integers $j\geq 1$ such that $\theta'_j<\theta'_{j+1}$.
In the following, we will give the algebraic form of formula (\ref{qQ}).
\begin{prp}
Let Young diagram $\mu=(\mu_1,\cdots,\mu_l)$ and $\mu_{i+1}< r\leq\mu_i$ for $0\leq i\leq l$ with $\mu_{l+1}=0,\ \mu_{0}=+\infty$,
\be
Q_\mu({\bf x}) q_r({\bf x})=\prod_{j=1}^i\frac{1-tR_{j(i+1)}}{1-R_{j(i+1)}}\prod_{k=i+2}^{l+1}\frac{1-tR_{(i+1)k}}{1-R_{(i+1)k}}\cdot Q_{\mu\cup(r)}({\bf x}).
\ee
\end{prp}
\begin{proof}
From equation (\ref{qdef}) and lemma \ref{Qqlem},
\begin{eqnarray*}
Q_\mu({\bf x}) q_r({\bf x})&=&[z_1^{\mu_1}\cdots z_l^{\mu_l}]X^+(z_1)\cdots X^+(z_l)\cdot 1 \ [w^r]X^+(w)\cdot 1\\
&=&[z_1^{\mu_1}\cdots z_i^{\mu_i}w^r z_{i+1}^{\mu_{i+1}}\cdots z_l^{\mu_l}]\prod_{j=1}^i\frac{1-tw/z_j}{1-w/z_j}\prod_{j=i+1}^l\frac{1-tz_j/w}{1-z_j/w}\\
&&V^+(z_1)\cdots V^+(z_i)V^+(w) V^+( z_{i+1})\cdots V^+(z_l)\cdot 1\\
&=&\prod_{j=1}^i\frac{1-tR_{j(i+1)}}{1-R_{j(i+1)}}\prod_{k=i+2}^{l+1}\frac{1-tR_{(i+1)k}}{1-R_{(i+1)k}}\cdot Q_{\mu\cup(r)}({\bf x}).
\end{eqnarray*}
\end{proof}
\begin{eg} Let $\mu=(2,1)$ and $r=2$,
\begin{eqnarray*}
&&Q_{(2,1)}({\bf x})q_2({\bf x})=\frac{1-tR_{12}}{1-R_{12}}\frac{1-tR_{23}}{1-R_{23}}\cdot Q_{(2,2,1)}({\bf x})\\
&=& (1+(1-t)R_{12}+(1-t)R_{12}^2+(1-t)R_{12}^3+(1-t)R_{12}^4)(1+(1-t)R_{23})\cdot Q_{(2,2,1)}({\bf x})\\
&=& Q_{(2,2,1)}({\bf x})+(1-t)Q_{(3,1,1)}({\bf x})+(1-t)Q_{(4,0,1)}({\bf x})+(1-t)Q_{(5,-1,1)}({\bf x})\\
&&+(1-t)Q_{(6,-2,1)}+(1-t)Q_{(2,3)}({\bf x})+(1-t)^2({\bf x})Q_{(3,2)}({\bf x})+(1-t)^2Q_{(4,1)}({\bf x})\\
&&+(1-t)^2Q_{(5)}({\bf x})+(1-t)^2Q_{(6,-1)}({\bf x})\\
&=&Q_{(2,2,1)}({\bf x})+(1-t)Q_{(3,1,1)}({\bf x})+(1-t)Q_{(3,2)}({\bf x})+(1-t)Q_{(4,1)}({\bf x}).\end{eqnarray*}
where the second equation holds since $(\mu\cup(r))_2=2$ and $(\mu\cup(r))_3=1$ and the last equation hold since
\begin{eqnarray*}
Q_{(5,-1,1)}=(t-1)Q_{(5)},\ \ Q_{(4,0,1)}=tQ_{(4,1)},\ \ Q_{(2,3)}=tQ_{(3,2)}\ \ Q_{(6,-2,1)}=0
\end{eqnarray*}
by using the formulas in propositions \ref{prp23} and \ref{prp24} and
\[
Q_{(6,-1)}=0
\]
by lemma \ref{qal12}.
\end{eg}

In the following, we give the algebraic form of the decomposition formula for $Q_{\mu}({\bf x})Q_{\nu}({\bf x})$.
\begin{prp}
Let $\mu=(\mu_1,\mu_2,\cdots,\mu_l)$ and $\nu=(\nu_1,\nu_2,\cdots,\nu_{l'})$. Suppose $\mu_{i_k}\geq \nu_k>\mu_{i_k+1}$ for $k=1,2,\cdots,l'$ with $\mu_0=+\infty,\ \mu_{l+1}=0$. Denote the set $\{i_1+1,i_2+2,\cdots, i_{l'}+l'\}$ by $I$, then the following formula holds
\be
Q_\mu({\bf x})Q_\nu({\bf x})=\prod_{k=1}^{l'}\prod_{j=1,j\not\in I}^{i_k+k-1}\frac{1-tR_{j(i_k+k)}}{1-R_{j(i_k+k)}}\prod_{j=i_k+k+1,j\not\in I}^{l+l'}\frac{1-tR_{j(i_k+k)}}{1-R_{j(i_k+k)}}\cdot Q_{\mu\cup\nu}({\bf x}).
\ee
\end{prp}
\begin{proof}
From equation (\ref{slambda}) and lemma \ref{vv},
\begin{eqnarray*}
&&Q_\mu({\bf x}) Q_\nu({\bf x})=[{\bf z}^\lambda]X^+(z_1)\cdots X^+(z_l)\cdot 1 \ [{\bf w}^\mu]X^+(w_1)\cdots X^+(w_{l'})\cdot 1\\
&=&[{\bf z}^\lambda{\bf w}^\mu]\prod_{k=1}^{l'}\prod_{j=1}^{i_k}\frac{1-tw_{i_k+k}/z_j}{1-w_{i_k+k}/z_j}\prod_{j=i_k+1}^{l}\frac{1-tz_{j}/w_{i_k+k}}{1-z_{j}/w_{i_k+k}}
X^+(z_1)\cdots X^+(z_{i_1})\\&&X^+(w_1) X^+( z_{i_1+1})\cdots X^+(z_{i_{l'}})X^+(w_{l'}) X^+( z_{i_{1'+1}})\cdots X^+(z_l)\cdot 1\\
&=&\prod_{k=1}^{l'}\prod_{j=1,j\not\in I}^{i_k+k-1}\frac{1-tR_{j(i_k+k)}}{1-R_{j(i_k+k)}}\prod_{j=i_k+k+1,j\not\in I}^{l+l'}\frac{1-tR_{j(i_k+k)}}{1-R_{j(i_k+k)}}\cdot Q_{\mu\cup\nu}({\bf x}).
\end{eqnarray*}
\end{proof}
\begin{eg}Let $\mu=\nu=(2,1)$,
\begin{eqnarray*}
&&Q_\mu({\bf x}) Q_\nu({\bf x})=\frac{1-tR_{12}}{1-R_{12}}\frac{1-tR_{23}}{1-R_{23}}\frac{1-tR_{14}}{1-R_{14}}\frac{1-tR_{34}}{1-R_{34}}\cdot Q_{(2,2,1,1)}({\bf x})\\
&=& \frac{1-tR_{12}}{1-R_{12}}\frac{1-tR_{23}}{1-R_{23}}(1+(1-t)R_{14})(1+(1-t)R_{34})\cdot Q_{(2,2,1,1)}({\bf x})\\
&=& \frac{1-tR_{12}}{1-R_{12}}\frac{1-tR_{23}}{1-R_{23}}\left(Q_{(2,2,1,1)}+(1-t)Q_{(3,2,1)}+(1-t)Q_{(2,2,2)}\right)\\
&=& \frac{1-tR_{12}}{1-R_{12}}\left(Q_{(2,2,1,1)}+(1-t)Q_{(3,2,1)}+(1-t)Q_{(2,2,2)}+(1-t)^2Q_{(3,3)}+(1-t)Q_{(2,3,1)}\right)\\
&=&Q_{(2,2,1,1)}({\bf x})+(1-t)^2Q_{(3,3)}({\bf x})+(1-t^2)(2-t)Q_{(3,2,1)}({\bf x})+(1-t)Q_{(2,2,2)}({\bf x})\\
&&+(1-t)Q_{(3,1,1,1)}({\bf x})+(1-t)Q_{(4,1,1)}({\bf x})+(1-t)^2Q_{(4,2)}({\bf x}).
\end{eqnarray*}
\end{eg}
\section{Universal characters}\label{sect4}
Let ${\bf x}=(x_1,x_2,\cdots)$ and ${\bf y}=(y_1,y_2,\cdots)$. For Young diagrams $\lambda=(\lambda_1,\lambda_2,\cdots,\lambda_l)$ and $\mu=(\mu_1,\mu_2,\cdots,\mu_{l'})$, the universal character $S_{[\lambda,\mu]}=S_{[\lambda,\mu]}({\bf x},{\bf y})$ is a polynomial of variables ${\bf x}$ and ${\bf y}$ in $\C[{\bf x},{\bf y}]$ defined by the twisted Jacobi-Trudi formula \cite{KK}:
\be
S_{[\lambda,\mu]}({\bf x},{\bf y})=\text{det}\left( \begin{array}{cc}
h_{\mu_{l'-i+1}+i-j}({\bf y}), & 1\leq i\leq l' \\
h_{\lambda_{i-l'}-i+j}({\bf x}), & l'+1\leq i\leq l+l'
\end{array} \right)_{1\leq i,j\leq l+l'}.
\ee
Define the degree of each variables $x_n, y_n, \ n=1,2,\cdots$ by
\[
\text{deg }x_n=n,\quad \text{deg }y_n=-n
\]
then $S_{[\lambda,\mu]}({\bf x},{\bf y})$ is a homogeneous polynomial of degree $|\lambda|-|\mu|$, where $|\lambda|=\lambda_1+\lambda_2+\cdots+\lambda_l$ is called the weight of $\lambda$. Note that $S_\lambda({\bf x})$ is a special case of the universal character: $S_\lambda({\bf x})=\det(h_{\lambda_i-i+j}({\bf x}))=S_{[\lambda,\emptyset]}({\bf x},{\bf y})$.

Introduce the following vertex operators
\bea
&&\Gamma_1^-(z)=e^{\xi({\bf x}-\tilde{\partial}_{\bf y}, z)},\quad \Gamma_1^+(z)=e^{\xi(\tilde{\partial}_{\bf x}, z^{-1})},\label{g1}\\
&&\Gamma_2^-(z)=e^{\xi({\bf y}-\tilde{\partial}_{\bf x}, z)},\quad \Gamma_2^+(z)=e^{\xi(\tilde{\partial}_{\bf y}, z^{-1})}.\label{g2}
\eea
Define
\bea
X^\pm(z)&=&\sum_{n\in\Z}X^\pm_nz^n=e^{\pm\xi({\bf x}-\tilde{\partial}_{\bf y}, z)}e^{\mp\xi(\tilde{\partial}_{\bf x}, z^{-1})},\\
Y^\pm(z^{-1})&=&\sum_{n\in\Z}Y^\pm_nz^{-n}=e^{\pm\xi({\bf y}-\tilde{\partial}_{\bf x}, z^{-1})}e^{\mp\xi(\tilde{\partial}_{\bf y}, z)}.
\eea
The operators $X_i^\pm$ satisfy the following Fermionic relations:
\bea
X_i^\pm X_j^\pm +X_{j-1}^\pm X_{i+1}^{\pm}&=&0,\nonumber\\
X_i^+X_j^-+X_{j+1}^{-}X_{i-1}^+&=&\delta_{i+j,0}.\nonumber
\eea
The same relations hold also for $Y_i^\pm$, and $X_i^\pm$ and $Y_i^\pm$ are commutative. The operators $X_i^+$ and $Y_i^+$ are raising operators for the universal characters such that
\begin{eqnarray}
S_{[\lambda,\mu]}({\bf x},{\bf y})&=&X_{\lambda_1}^+\cdots X_{\lambda_l}^+Y_{\mu_1}^+\cdots X_{\mu_{l'}}^+\cdot 1\nonumber\\
&=&[z^\lambda w^\mu] X^+(z_1)\cdots X^+(z_l)Y^+(w_1)\cdots Y^+(w_{l'})\cdot 1\label{universal}
\end{eqnarray}
where the Young diagrams $\lambda=(\lambda_1,\lambda_2,\cdots,\lambda_l)$ and $\mu=(\mu_1,\mu_2,\cdots,\mu_{l'})$. It turns out that
\be\label{slambdamu}
S_{[\lambda,\mu]}({\bf x},{\bf y})=S_\lambda({\bf x}-\tilde{\partial}_{\bf y})S_\mu({\bf y}-\tilde{\partial}_{\bf x})\cdot 1
\ee
where $S_\lambda({\bf x})$ is a Schur function.

From \cite{KK}, we know that
\be\label{ssUC}
S_{[\xi,\eta]}({\bf x},{\bf y})S_{[\tau,\mu]}({\bf x},{\bf y})=\sum_{\lambda,\mu} M_{[\xi,\eta],[\tau,\nu]}^{[\lambda,\mu]}S_{[\lambda,\mu]}({\bf x},{\bf y}),
\ee
where\be
 M_{[\xi,\eta],[\tau,\nu]}^{[\lambda,\mu]}=\sum_{\alpha,\beta,\theta,\delta}(\sum_\kappa C_{\kappa\alpha}^\xi C_{\kappa\beta}^\nu)(\sum_\epsilon C_{\epsilon\theta}^\eta C_{\epsilon\delta}^\tau) C_{\alpha\delta}^\lambda C_{\beta\theta}^\mu
\ee
and $C_{\mu\nu}^\lambda$ is the Littlewood-Richardson coefficient of $S_\lambda({\bf x})$ in the expansion of $S_\mu({\bf x})S_\nu({\bf x})$.
In the following, we will give the algebraic form of the formula (\ref{ssUC}).

For $\lambda=(\lambda_1,\lambda_2,\cdots,\lambda_l)$, define the operator $D_{\lambda_i}$ by $D_{\lambda_i}\cdot \lambda=(\lambda_1,\cdots,\lambda_{i-1},\lambda_i-1,\lambda_{i+1},\cdots,\lambda_{l})$, and $D_{\lambda_i\mu_j}=D_{\lambda_i}D_{\mu_j}$.
\begin{prp}For Young diagrams $\xi,\ \eta,\ \tau,\ \nu$,
\be
S_{[\xi,\eta]}({\bf x},{\bf y})S_{[\tau,\nu]}({\bf x},{\bf y})=\prod_{ij}\frac{1}{1-D_{\xi_i\nu_j}}\prod_{mn}\frac{1}{1-D_{\tau_m\eta_n}}S_{[\xi\cdot\tau,\eta\cdot\nu]}({\bf x},{\bf y})
\ee
where the multiplication $\mu\cdot\nu=\sum_{\lambda}C_{\mu\nu}^\lambda \lambda$ satisfies the Littlewood-Richardson rule.
\end{prp}
\begin{proof}
This can be proved by using the formula (\ref{slambdamu}) and the vertex operator realization of universal character (\ref{universal}).
\end{proof}
\begin{eg}
Let $\xi=(2,1),\ \eta=(3,1),\ \tau=(1),\ \nu=(1)$,
\begin{eqnarray*}
&&S_{[\xi,\eta]}({\bf x},{\bf y})S_{[\tau,\nu]}({\bf x},{\bf y})=\prod_{i}\frac{1}{1-D_{\xi_i\nu}}\prod_{n}\frac{1}{1-D_{\tau\eta_n}}S_{[\xi\cdot\tau,\eta\cdot\nu]}({\bf x},{\bf y})\\
&=&(1+D_{\xi_1\nu})(1+D_{\xi_2\nu})(1+D_{\tau\eta_1})(1+D_{\tau\eta_2})S_{[(2,1)\cdot(1),(3,1)\cdot(1)]}({\bf x},{\bf y})\\
&=& (1+D_{\xi_1\nu}+D_{\xi_2\nu}+D_{\tau\eta_1}+D_{\tau\eta_2}+D_{\xi_1\nu}D_{\tau\eta_1}+D_{\xi_2\nu}D_{\tau\eta_1}+D_{\xi_1\nu}D_{\tau\eta_2}\\
&&+D_{\xi_2\nu}D_{\tau\eta_2})S_{[(2,1)\cdot(1),(3,1)\cdot(1)]}({\bf x},{\bf y})\\
&=&S_{[(2,1)\cdot(1),(3,1)\cdot(1)]}+S_{[(1,1)\cdot(1),(3,1)]}+S_{[(2)\cdot(1),(3,1)]}+S_{[(2,1),(2,1)\cdot(1)]}\\
&&+S_{[(2,1),(3)\cdot(1)]}+S_{[(1,1),(2,1)]}+S_{[(2),(2,1)]}+S_{[(1,1),(3)]}+S_{[(2),(3)]}\\
&=&S_{[(3,1),(4,1)]}+S_{[(3,1),(3,2)]}+S_{[(3,1),(3,1,1)]}+S_{[(2,2),(4,1)]}+S_{[(2,2),(3,2)]}+S_{[(2,1),(3,1,1)]}\\
&&+S_{[(2,1,1),(4,1)]}+S_{[(2,1,1),(3,2)]}+S_{[(2,1,1),(3,1,1)]}+S_{[(2,1),(3,1)]}+S_{[(1,1,1),(3,1)]}\\
&&S_{[(3),(3,1)]}+S_{[(2,1),(3,1)]}+S_{[(2,1),(3,1)]}+S_{[(2,1),(2,2)]}+S_{[(2,1),(2,1,1)]}+S_{[(2,1),(4)]}\\
&&+S_{[(2,1),(3,1)]}+S_{[(1,1),(2,1)]}+S_{[(2),(2,1)]}+S_{[(1,1),(3)]}+S_{[(2),(3)]},
\end{eqnarray*}
where we use $S_{[\lambda,\mu]}$ to denote $S_{[\lambda,\mu]}({\bf x},{\bf y})$.
\end{eg}

\section*{Discussion}
From the formulas we constructed, it is easy to find that the coefficients of $S_\lambda({\bf x})$ in the expansion of $S_\mu({\bf x})S_\nu({\bf x})$ and the coefficients of $S_{[\lambda,\mu]}({\bf x},{\bf y})$ in the expansion of $S_{[\xi,\eta]}({\bf x},{\bf y})S_{[\tau,\mu]}({\bf x},{\bf y})$ are integers, and the coefficients of $Q_\lambda({\bf x})$ in the expansion of $Q_\mu({\bf x})Q_\nu({\bf x})$ are in $\Z[t]$. We can use this method to prove some structure constants are integers.
\section*{Acknowledgements}
The authors gratefully acknowledge the support of Professors Shi-Kun Wang, Zi-Feng Yang.
Ke Wu is supported by the National Natural Science Foundation
of China under Grant No. 11475116.

\end{document}